\RequirePackage{fix-cm}
\documentclass[smallextended]{svjour3}       
\smartqed  
\usepackage{graphicx}
 \usepackage{mathptmx}      
\usepackage{amsmath}
\newtheorem{rem}{Remark}
\usepackage{amssymb}

\begin{document}

\title{Further results on the 2-adic complexity of a class of balanced generalized cyclotomic sequences
\thanks{This work is financially supported by the National Natural Science Foundation of China (No. 61902429), Fundamental Research Funds for the Central Universities (No. 19CX02058A), Shandong Provincial Natural Science Foundation of China (ZR2019MF070).\\
$^{\ast}$\ The corresponding author: zhaochune1981@163.com}}


\titlerunning{Autocorrelation distribution and $2$-adic complexity}        

\author{Chun-e Zhao$^{\ast}$\and Yuhua Sun \and Tongjiang Yan}


\institute{Chun-e Zhao, Yuhua Sun, Tongjiang Yan
\at
College of Sciences,
China University of Petroleum,
Qingdao 266555,
Shandong, China
}

\date{Received: date / Accepted: date}

\maketitle

\begin{abstract}
In this paper, the 2-adic complexity of a class of balanced Whiteman generalized cyclotomic sequences of period $pq$ is considered. Through calculating the determinant of the circulant matrix constructed by one of these sequences, we derive a lower bound on the 2-adic complexity of the corresponding sequence, which further expands our previous work (Zhao C, Sun Y and Yan T. Study on 2-adic complexity of a class of balanced generalized cyclotomic sequences. Journal of Cryptologic Research,6(4):455-462, 2019). The result shows that the 2-adic complexity of this class of sequences is large enough to resist the attack of the rational approximation algorithm(RAA) for feedback with carry shift registers(FCSRs), i.e., it is in fact lower bounded by $pq-p-q-1$, which is far larger than one half of the period of the sequences. Particularly, the 2-adic complexity is maximal if suitable parameters are chosen.
\keywords{ 2-adic complexity \and generalized cyclotomic sequences \and  circulant matrix }
\end{abstract}

\section{Introduction}
The concept of the 2-adic complexity of a binary sequence was originally considered by Klapper and Goresky , meanwhile, the rational approximation algorithm (RAA) to attack a given key stream sequence with low 2-adic complexity was also presented by them \cite{Andrew Klapper}. Therefore, the 2-adic complexity of a key stream sequence has become an important security criterion and it should be no less than one half of the period of the sequence according to RAA.

In the past few decades, many generalized cyclotomic sequences with high linear complexity have been constructed and have been widely used in communication systems.
However, the 2-adic complexity of only a small part of these sequences has been completely determined. For example, the 2-adic complexity of twin-prime sequences is maximal which has been proved as a class of sequences with ideal 2-level autocorrelation by Xiong et al. \cite{Xionghai-1}. And the 2-adic complexity of modified Jacobi sequences of period $pq$ was originally showed to satisfy a lower bound $pq-p-q-1$ by Sun et al. \cite{SunY} and then was proved to be maximal in the case of $\frac{q+1}{4}\leq p\leq 4q-1$ by Hofer and Winterhof \cite{Hofer}. Moreover, the 2-adic complexity of another two classes of Whiteman generalized cyclotomic sequences of period $pq$ has also been proved to have large 2-adic complexity by Zeng et al. \cite{Zeng X} and Sun et al. \cite{Sunyuhua-1}. It should be pointed out that all the above sequences are unbalanced.

In this paper, we aim to analyze the 2-adic complexity of a class of balanced Whiteman generalized cyclotomic sequences whose linear complexity has been proved to be large enough to resist B-M algorithm by Bai et al. \cite{Bai}. Our result shows that the 2-adic complexity of these sequences is also large enough and it is even optimal when the parameters are suitable chosen.

The rest of the paper is organized as follows. We give the relevant definitions of Whiteman generalized cyclotomic sequences and some known results Whiteman generalized cyclotomic classes in Section 2. In Section 3, by means of a new property of Gaussian periods and the method of Xiong et al. \cite{Xionghai-1}, we study the 2-adic complexity of  a class of balanced Whiteman generalized cyclotomic sequences and give a nontrivial lower bound on the 2-adic complexity. Finally, we summarize our results in Section 4.

\section{The relevant definitions of a class of balanced Whiteman generalized cyclotomic sequences and some known results of Gaussian periods}
Firstly, we give some definitions and symbols which will be always used in the whole paper unless otherwise specified.

Let $p,q$ be two different primes satisfying $\mathrm{gcd}(p-1,q-1)=2$ and $g$ be a common primitive root of $p$ and $q$. Denote $N=pq$. By China Remainder Theorem, the system of the congruence equations
$$\left\{
\begin{array}{cc}
  x\equiv g & \pmod p, \\
 x\equiv1 & \pmod q
\end{array}
\right.$$
has only one solution $x$ in the ring $Z_{N}$ of residue classes modulo $N$.
Denote $e=(p-1)(q-1)/2$. Then the Whiteman generalized cyclotomic classes $D_{i}$ of order 2 with resect to $p$ and $q$ are defined as
$$D_{i}=\{g^{s}x^{i}\pmod N:s=0,1,\cdots,e-1\},\ \ i=0,1.$$
Let $Z_N^{\ast}$ be the multiplicative group of integers modulo $N$.
Whiteman has showed that
$$Z_{N}^{*}=D_{0}\cup D_{1}, D_{0}\cap D_{1}=\emptyset,$$
where $\emptyset$ is the empty set. The cyclotomic numbers associated with the Whiteman generalized cyclotomic classes of order 2 with respect to $p$ and $q$ are defined by
$$(i,j)=|(D_{i}+1)\cap D_{j}|,$$
where $D_{i}+1=\{x+1|x\in D_{i}\}$, $i=0,1$, $j=0,1$. Let
 \begin{align*}
D_{0}^{(p)}=&\{g^{2t}\pmod p:t=0,1,\cdots,\frac{p-1}{2}-1\},\\
D_{1}^{(p)}=&\{g^{2t+1}\pmod p:t=0,1,\cdots,\frac{p-1}{2}-1\},\\
D_{0}^{(q)}=&\{g^{2t}\pmod q:t=0,1,\cdots\frac{q-1}{2}-1\},\\
D_{1}^{(q)}=&\{g^{2t+1}\pmod q:t=0,1,\cdots,\frac{q-1}{2}-1\}.
 \end{align*}
Then $$P=\{p,2p,\cdots,(q-1)p\}=D_{0}^{(q)}p\cup D_{1}^{(q)}p,$$
$$Q=\{q,2q,\cdots,(p-1)q\}=D_{0}^{(p)}q\cup D_{1}^{(p)}q,$$
where $D_{i}^{(p)}q=\{xq|x\in D_{i}^{(p)}\}$, $D_{i}^{(q)}p=\{xp|x\in D_{i}^{(q)}\}$, $i=0,1.$

Let
$$C_{0}=\{0\}\cup D_{0}\cup D_{0}^{(p)}q\cup D_{0}^{(q)}p,\ C_{1}=D_{1}\cup D_{1}^{(p)}q\cup D_{1}^{(q)}p.$$
Then
$$Z_{N}=C_{0}\cup C_{1}.$$
 Define the sequence
$s=\{s_{i}\}_{i=0}^{N-1}$ of period $N$ as

\begin{equation}\label{eq1}
 s_{i}=\left\{\begin{array}{cc}
            0, & i \pmod N\in C_{0}, \\
            1, & i \pmod N\in C_{1},
          \end{array}\right.
\end{equation}
which was originally given by Ding and Helleseth\cite{Ding C} and has been proved by Bai et al. to have high linear complexity \cite{Bai}.

For an arbitrary binary sequence $s=\{s_{i}\}_{i=0}^{N-1}$ of period $N$, let $S(x)=\sum\limits_{i=0}^{N-1}s_{i}x^{i}\in Z[x]$. If
 \begin{equation}\label{eq2}
 \frac{S(2)}{2^{N}-1}=\frac{\sum\limits_{i=0}^{N-1}s_{i}2^{i}}{2^{N}-1}=\frac{m}{n},\ 0\leq m\leq n,\ \mathrm{gcd}(m,n)=1,
 \end{equation}
 then the 2-adic complexity $\varphi_{2}(s)$ of $s$ is defined by $\lfloor \mathrm{log}_{2}(n+1)\rfloor$, where $\mathrm{gcd}(m,n)$ is the greatest common divisor of the integers $m,n$ and $\lfloor \mathrm{log}_{2}n\rfloor$ is the maximal integer no more than $\mathrm{log}_{2}n$.

 From Eq. \eqref{eq2}, we know that $\varphi_{2}(s)$ can be calculated by
\begin{equation}\label{eq3}
  \varphi_{2}(s)=\Big\lfloor \mathrm{log}_{2}\Big(\frac{2^{N}-1}{\mathrm{gcd}\big(2^{N}-1,S(2)\big)}+1\Big)\Big\rfloor.
\end{equation}

Let $\omega_{p}=e^{\frac{2\pi i}{p}}$ be a complex $p$th primitive root of unity and $\omega_{N}=e^{\frac{2\pi i}{N}}$ be a complex $N$th primitive root of unity, then $\delta_{i}^{p}=\sum\limits_{k\in D_{0}^{(p)}}\omega_{p}^{k}$ and $\eta_{i}=\sum\limits_{k\in D_{i}}\omega_{N}^{k}$ are called Gauss periods based on classical cyclotomic classes $D_{i}^{(p)}$ and generalized cyclotomic classes $D_{i}$ respectively, where $i=0,1$.

In order to give our main result, we need the following results of Gaussian periods.
\begin{lemma}\cite{SunY}\label{SunY}
Let $p,q$ be two distinct primes with $\mathrm{gcd}(p-1,q-1)=2$. Then
\begin{itemize}
\item[(1)] $\eta_{0}+\eta_{1}=1,\eta_{0}\eta_{1}=\frac{1+pq}{4}$;
\item[(2)] for an odd prime $p\equiv1\pmod4$, $\delta_{0}^{p}+\delta_{1}^{p}=-1,\delta_{0}^{p}\delta_{1}^{p}=\frac{1-p}{4}$;
\item[(3)] for an odd prime $p\equiv3\pmod4$, $\delta_{0}^{p}+\delta_{1}^{p}=-1,\delta_{0}^{p}\delta_{1}^{p}=\frac{1+p}{4}$.
\end{itemize}
\end{lemma}

\section{A lower bound on the 2-adic complexity of the balanced Whiteman generalized cyclotomic sequences}
In this section, using the method of Xiong et al.\cite{Xionghai-1}, a lower bound on the 2-adic complexity of the balanced Whiteman generalized cyclotomic sequences defined in \eqref{eq1} for the case of $|q-p|<\sqrt{pq}-1$. To this end, we list the following lemmas in turn.

\begin{lemma}
Let $p\equiv1\pmod4,\ q\equiv3\pmod4$, then $$\eta_{0}^{2}(\delta_{1}^{p}\delta_{1}^{q}+\delta_{0}^{p}\delta_{0}^{q})+\eta_{1}^{2}(\delta_{1}^{p}\delta_{0}^{q}+\delta_{0}^{p}\delta_{1}^{q})=\frac{1-pq}{4}\pm\frac{pq}{2}.
$$
\end{lemma}
\begin{proof}
 By lemma \ref{SunY} (2)-(3), we have $$\delta_{1}^{p}=\frac{-1\pm\sqrt{p}}{2},\ \delta_{0}^{p}= \frac{-1\mp\sqrt{p}}{2},$$
 $$\delta_{1}^{q}=\frac{-1\pm\sqrt{q}i}{2},\ \delta_{0}^{q}= \frac{-1\mp\sqrt{q}i}{2}.$$
Then, by direct calculation, we get
 $$\delta_{1}^{p}\delta_{1}^{q}+\delta_{0}^{p}\delta_{0}^{q}=\frac{1\pm\sqrt{pq}i}{2},\ \delta_{1}^{p}\delta_{0}^{q}+\delta_{0}^{p}\delta_{1}^{q}=\frac{1\mp\sqrt{pq}i}{2}.$$
Again, by Lemma \ref{SunY} (1) and direct calculation, we can further obtain
\begin{align*}
&\eta_{0}^{2}(\delta_{1}^{p}\delta_{1}^{q}+\delta_{0}^{p}\delta_{0}^{q})+\eta_{1}^{2}(\delta_{1}^{p}\delta_{0}^{q}+\delta_{0}^{p}\delta_{1}^{q})\\
=&\eta_{0}^{2}\times\frac{1+\sqrt{pq}i}{2}+\eta_{1}^{2}\times\frac{1-\sqrt{pq}i}{2}\\
=&\frac{1}{2}(\eta_{0}^{2}+\eta_{1}^{2})+\frac{\sqrt{pq}i}{2}(\eta_{0}^{2}-\eta_{1}^{2})\\
=&\frac{1-pq}{4}\pm\frac{pq}{2}.
\end{align*}
\end{proof}

\begin{lemma}\label{main method-2}\cite{Davis},\cite{Xionghai-1}
Let $s=\{s_{i}\}_{i=0}^{N-1}$ be a binary sequence of period $N$ and denote $S(x)=\sum\limits_{i=0}^{N-1}s_{i}x^{i}\in Z[x]$. Suppose $A=(a_{i,j})_{N\times N}$ is the matrix defined by $a_{i,j}=s_{i-j\pmod N}$. Viewing $A$ as a matrix over the rational field $\mathbb{Q}$. Then
\begin{itemize}
\item[(1)] $\mathrm{det}(A)=\prod_{a=0}^{N-1}S(\omega_N^a)$;
\item[(2)] If the determinant of $A$ satisfies $\mathrm{det}(A)\neq0$, then
\begin{align}
 \mathrm{gcd}\big(S(2),2^N-1\big)\mid \mathrm{gcd}\big(\mathrm{det}(A),2^N-1\big).\label{method}
\end{align}
\end{itemize}
\end{lemma}

From Lemma \ref{main method-2}, to get a lower bound of the 2-adic complexity of a binary sequence we need to determine the value of $\mathrm{det}(A)$ and the value of $\mathrm{det}(A)$ can be obtained by $\prod_{a=0}^{N-1}S(\omega_N^a)$. The following three lemmas provide the desired results.

\begin{lemma}\cite{Yan}\label{Yan}
For any $a\in Z_{N}$ and $B\subset Z_{N}$, denote $aB=\{ab\pmod N|b\in B\}$, then there are following properties .
\begin{itemize}
\item[(1)]For every fixed $a\in D_{i}$, $aP=P,aQ=Q$ and $aD_{j}=D_{(i+j)\pmod2}$, where $i,j=0,1$.
\item[(2)]For every fixed $a\in P$, if $b$ runs through $D_{i}$, then $ab$ runs though $P$ every elements $\frac{p-1}{2}$ times, and $aP=P,aQ=R$.
\item[(3)]For every fixed $a\in Q$, if $b$ through $D_{i}$, then $ab$ runs through $Q$ every elements $\frac{q-1}{2}$ times. and $aQ=Q,aP=R$.
\item[(4)]For every fixed $a\pmod p \in D_{i}^{(p)}$, then $aD_{j}^{(p)}=D_{i+j\pmod2}^{(p)},i,j=0,1$.
\item[(5)]For every fixed $a\pmod p \in D_{i}^{(q)}$, then $aD_{j}^{(q)}=D_{i+j\pmod2}^{(q)},i,j=0,1$.
\end{itemize}
\end{lemma}

In convenience, we denote
\begin{align*}
D_{00}=&\{g^{2t}\pmod N:t=0,1,\cdots,\frac{e}{2}-1\},\\
D_{01}=&\{g^{2t+1}\pmod N:t=0,1,\cdots,\frac{e}{2}-1\},\\
D_{10}=&\{g^{2t}x\pmod N:t=0,1,\cdots,\frac{e}{2}-1\},\\
D_{11}=&\{g^{2t+1}x\pmod N:t=0,1,\cdots,\frac{e}{2}-1\},
\end{align*}
i.e., $D_{0}=D_{00}\cup D_{01}$ and $D_{1}=D_{10}\cup D_{11}$.
\begin{lemma}\label{Sproduct}
Let $\{s_{i}\}_{i=0}^{N-1}$ be the sequence defined in Eq.(\ref{eq1}), then
\begin{center}
$S(\omega_{N}^{a})=\left\{\begin{array}{cc}
                            \frac{pq-1}{2} & a\in R \\
                            \delta_{1}^{p}  & \ \ \ \ \ \ a\in D_{0}^{(q)}p \\
                            \delta_{0}^{p}  & \ \ \ \ \ \ a\in D_{1}^{(q)}p \\
                            \delta_{1}^{q}  & \ \ \ \ \ \ a\in D_{0}^{(p)}q \\
                            \delta_{0}^{q}  & \ \ \ \ \ \ a\in D_{1}^{(p)}q \\
                            \eta_{1}+\delta_{0}^{q}+\delta_{0}^{p}  & \ \ \ a\in D_{01} \\
                           \eta_{1}+\delta_{1}^{q}+\delta_{1}^{p}   &\ \  \ a\in D_{00} \\
                            \eta_{0}+\delta_{0}^{q}+\delta_{1}^{p}  & \ \ \  a\in D_{10} \\
                            \eta_{0}+\delta_{1}^{q}+\delta_{0}^{p}  &  \ \ \ a\in D_{11}
                          \end{array}
\right.$
\end{center}
\end{lemma}
\begin{proof}
We need only apply the results of Lemma \ref{Yan} to the definition of $S(\omega_{N}^{a})$ according to the set to which $a$ belongs. For the sake of brevity, here we omit the details.
\end{proof}
\begin{lemma}\label{mainlemma}
Let $s=(s_{0},s_{1},\cdots,s_{N-1})$ be the binary sequence defined in Eq.(\ref{eq1}) and $A=(a_{i,j})_{N\times N}$ be the matrix defined by $a_{i,j}=s_{(i-j)\pmod N}$. Then
\begin{center}
$\mathrm{det}(A)=(\frac{pq-1}{2})(\frac{1-p}{4})^{\frac{p-1}{2}}(\frac{1+q}{4})^{\frac{q-1}{2}}\Delta^{\frac{e}{2}},$
\end{center}
where $\Delta=[\frac{(1+pq)^{2}}{16}+(\frac{\pm1-d}{2})pq+d^{2}+\frac{3}{2}d+\frac{1}{2}]$ and $d=\frac{q-p-2}{4}$.
\end{lemma}
\begin{proof}
By Lemmas \ref{main method-2} and \ref{Sproduct}, we have $\mathrm{det}(A)=\prod\limits_{a=0}^{N-1}S(\omega_{N}^{a})=(\frac{pq-1}{2})(\delta_{1}^{q}\delta_{0}^{q})^{\frac{q-1}{2}}(\delta_{1}^{p}\delta_{0}^{p})^{\frac{p-1}{2}}\Delta^{\frac{e}{2}}$,
here we temporarily denote $\Delta=(\eta_{1}+\delta_{1}^{q}+\delta_{1}^{p})(\eta_{1}+\delta_{0}^{q}+\delta_{0}^{p})(\eta_{0}+\delta_{0}^{q}+\delta_{1}^{p})(\eta_{0}+\delta_{1}^{q}+\delta_{0}^{p})$. By Lemmas 4 and 5, we know that $\delta_{1}^{q}\delta_{0}^{q}=\frac{1+q}{4}$ and $\delta_{1}^{p}\delta_{0}^{p}=\frac{1-p}{4}$. Next, we focus on $\Delta$.

\begin{align*}
\Delta=&(\eta_{1}+\delta_{1}^{q}+\delta_{1}^{p})(\eta_{1}+\delta_{0}^{q}+\delta_{0}^{p})(\eta_{0}+\delta_{1}^{q}+\delta_{0}^{p})(\eta_{0}+\delta_{0}^{q}+\delta_{1}^{p})\\
=&[\eta_{1}^{2}+(\delta_{0}^{q}+\delta_{0}^{p}+\delta_{1}^{q}+\delta_{1}^{p})\eta_{1}+(\delta_{1}^{q}+\delta_{1}^{p})(\delta_{0}^{q}+\delta_{0}^{p})][\eta_{0}^{2}+(\delta_{0}^{q}+\delta_{0}^{p}+\delta_{1}^{q}+\delta_{1}^{p})\eta_{0}\\
&+(\delta_{0}^{q}+\delta_{1}^{p})(\delta_{1}^{q}+\delta_{0}^{p})]\\
=&[\eta_{1}^{2}-2\eta_{1}+\frac{2+q-p}{4}+(\delta_{1}^{q}\delta_{0}^{p}+\delta_{1}^{p}\delta_{0}^{q})][\eta_{0}^{2}-2\eta_{0}+\frac{2+q-p}{4}+(\delta_{0}^{q}\delta_{0}^{p}+\delta_{1}^{q}\delta_{1}^{p})]\\
=&(A_{1}+B_{1})(A_{0}+B_{0})\\
=&A_{1}A_{0}+A_{0}B_{1}+A_{1}B_{0}+B_{1}B_{0},
\end{align*}
 where $A_{i}=\eta_{i}^{2}-2\eta_{i}+\frac{q-p+2}{4}$ and $B_{i}=\delta_{0}^{p}\delta_{i}^{q}+\delta_{1}^{p}\delta_{i+1\mod2}^{q}(i=0,1).$
Following, we analyze $A_{1}A_{0}, A_{0}B_{1}+A_{1}B_{0}$ and $B_{1}B_{0}$, respectively.\\

 \begin{align*}
A_{1}A_{0}=&[\eta_{1}^{2}-2\eta_{1}+\frac{q-p+2}{4}][\eta_{0}^{2}-2\eta_{0}+\frac{q-p+2}{4}]\ (by\ Lemma\ 3)\\
=&[\eta_{0}^{2}-1+\frac{q-p+2}{4}][\eta_{1}^{2}-1+\frac{q-p+2}{4}]\\
=&[\eta_{0}^{2}+\frac{q-p-2}{4}][\eta_{1}^{2}+\frac{q-p-2}{4}]\ (denote \frac{q-p-2}{4}\ by \  d)\\
=&\frac{(1+pq)^{2}}{16}+d[1-\frac{1+pq}{2}]+d^{2}.\\
A_{0}B_{1}+&A_{1}B_{0}\\
=&[\eta_{0}^{2}-2\eta_{0}+\frac{q-p+2}{4}][\delta_{0}^{p}\delta_{1}^{q}+\delta_{1}^{p}\delta_{0}^{q}]+[\eta_{1}^{2}-2\eta_{1}+\frac{q-p+2}{4}][\delta_{0}^{p}\delta_{0}^{q}+\delta_{1}^{p}\delta_{1}^{q}] \ (by \ Lemma \ 3)\\
=&[\eta_{1}^{2}+d][\delta_{0}^{p}\delta_{1}^{q}+\delta_{1}^{p}\delta_{0}^{q}]+[\eta_{0}^{2}+d][\delta_{0}^{p}\delta_{0}^{q}+\delta_{1}^{p}\delta_{1}^{q}]\\
=&\eta_{1}^{2}[\delta_{0}^{p}\delta_{1}^{q}+\delta_{1}^{p}\delta_{0}^{q}]+\eta_{0}^{2}[\delta_{0}^{p}\delta_{0}^{q}+\delta_{1}^{p}\delta_{1}^{q}]+d[\delta_{0}^{p}\delta_{1}^{q}+\delta_{1}^{p}\delta_{0}^{q}+\delta_{0}^{p}\delta_{0}^{q}+\delta_{1}^{p}\delta_{1}^{q}] \ (by\  Lemma \ 6)\\
=&\frac{1-pq}{4}\pm\frac{pq}{2}+d[(\delta_{0}^{p}+\delta_{1}^{p})(\delta_{0}^{q}+\delta_{1}^{q})] \ (by \ Lemmas \ 4 \ and \ 5)\\
=&\frac{1-pq}{4}\pm\frac{pq}{2}+d.
\end{align*}
\begin{align*}
B_{1}B_{0}=&[\delta_{0}^{p}\delta_{1}^{q}+\delta_{1}^{p}\delta_{0}^{q}][\delta_{0}^{p}\delta_{0}^{q}+\delta_{1}^{p}\delta_{1}^{q}]\\
=&\delta_{1}^{q}\delta_{0}^{q}(\delta_{0}^{p})^{2}+\delta_{0}^{p}\delta_{1}^{p}(\delta_{1}^{q})^{2}+\delta_{0}^{p}\delta_{1}^{p}(\delta_{0}^{q})^{2}+\delta_{1}^{q}\delta_{0}^{q}(\delta_{1}^{p})^{2}\ (by\ Lemmas \ 4 and \ 5)\\
=&[(\delta_{0}^{p})^{2}+(\delta_{1}^{p})^{2}]\delta_{1}^{q}\delta_{0}^{q}+[(\delta_{0}^{q})^{2}+(\delta_{1}^{q})^{2}]\delta_{1}^{p}\delta_{0}^{p}\\
=&(1-2\delta_{1}^{p}\delta_{0}^{p})\delta_{1}^{q}\delta_{0}^{q}+(1-2\delta_{1}^{q}\delta_{0}^{q})\delta_{1}^{p}\delta_{0}^{p}\\
=&(1-2\times\frac{1-p}{4})\times\frac{1+q}{4}+(1-2\times\frac{1+q}{4})\times\frac{1-p}{4}\\
=&\frac{1+pq}{4}
\end{align*}
Then $\Delta=[\frac{(1+pq)^{2}}{16}+(\frac{\pm1-d}{2})pq+d^{2}+\frac{3}{2}d+\frac{1}{2}]$.
\end{proof}

\begin{lemma}\cite{SunY}
Let $p$ and $q$ be different primes with $N=pq$, then $\mathrm{gcd}(2^{p}-1,\frac{2^{N}-1}{2^{p}-1})=\mathrm{gcd}(2^{p}-1,q)$ and $\mathrm{gcd}(2^{q}-1,\frac{2^{N}-1}{2^{q}-1})=\mathrm{gcd}(2^{q}-1,p)$. Especially, if $p<q$, then $\mathrm{gcd}(2^{q}-1,\frac{2^{N}-1}{2^{q}-1})=1.$
\end{lemma}
\begin{theorem}
Let $p$ and $q$ be two primes satisfying $p\equiv1\mod 4,q\equiv3\mod 4$ with $|q-p|<\sqrt{pq}-1$. Let $\{s_{i}\}_{i=0}^{N-1}$ be the Whiteman generalized cyclotomic sequence defined in Eq.(\ref{eq1}). Then the lower bound of the 2-adic complexity $\phi_{2}(s)$ is $pq-p-q-1$.
\end{theorem}
\begin{proof}
Let $r$ be a prime factor of $2^{N}-1$ and $\mathrm{Ord}_{r}(2)$ be the multiplicative order of 2 modulo $r$. Since $2^{N}-1=0\mod r, \mathrm{Ord}_{r}(2)|N$. So we get $\mathrm{Ord}_{r}(2)\in\{pq,p,q\}$. By Fermat's little theorem, we know that $2^{r-1}=1 \mod r$. Then $\mathrm{Ord}_{r}(2)|r-1$, therefore $r=k\mathrm{Ord}_{r}(2)+1$ where $k$ is a positive integer. From Lemma 7, we first calculate the value of $\mathrm{gcd}(\mathrm{det}(A)),2^{N}-1)$ for different cases.
\begin{itemize}
\item[(1)] $\mathrm{Ord}_{r}(2)=pq$.

In this case, we have $pq|r-1$, i.e., $r>pq$. Among the factors in $\mathrm{det}(A)$, $\frac{pq-1}{2},\frac{1-p}{4},\frac{1+q}{4}$ are less than $pq$. So we only need to analyze $\Delta$.

$(i)$ If $\Delta=[\frac{(1+pq)^{2}}{16}+\frac{1-d}{2}pq+d^{2}+\frac{3}{2}d+\frac{1}{2}]$, then there exist two integers $t$ and $k$ such that $r=kpq+1$ and $\Delta=[\frac{(1+pq)^{2}}{16}+\frac{1-d}{2}pq+d^{2}+\frac{3}{2}d+\frac{1}{2}]=t(kpq+1)$,i.e. $(1+pq)^{2}+8pq+16d^{2}+24d+8=16t(kpq+1)$ which is equivalent to
\begin{equation}\label{eq4}
pq(pq+10-8d)+(4d+3)^{2}=pq(16tk)+16t.
\end{equation}
By the fact that $|q-p|<\sqrt{pq}-1$, we have that $(4d+3)^{2}=(q-p+1)^{2}<pq$.
If $16t<pq$, then $16t=(4d+3)^{2}$ which is impossible for any integers $t$ and $d$. So $16t>pq$. By the eq.(\ref{eq4}), $16t>pq$ and $16kt<pq+10-8d$, again by the fact that $|q-p|<\sqrt{pq}-1$, then $8d=2(q-p-2)$ and  $16kt<pq+10-8d<pq+2\sqrt{pq}+16<2pq$. So $k=1$ and then $r=pq+1$ which contradicts with that $r$ is a prime. So $\mathrm{gcd}(r,\Delta)=1$.

$(ii)$ If $\Delta=[\frac{(1+pq)^{2}}{16}+\frac{1-d}{2}pq+d^{2}+\frac{3}{2}d+\frac{1}{2}]$, using the same method as that in the case $(i)$, we also have $\mathrm{gcd}(r,\Delta)=1$.

So, for the case of $\mathrm{Ord}_{r}(2)=pq$, we have $\mathrm{gcd}(r,\mathrm{det}(A))=1$.
\item[(2)] $\mathrm{Ord}_{r}(2)=p$.

 By Lemma 7, $\mathrm{gcd}(2^{p}-1,\frac{2^{N}-1}{2^{p}-1})=\mathrm{gcd}(2^{p}-1,q)$. $2^{r-1}=1\ \mathrm{mod}\ r$, So $p|r-1$ and then $r>p$. If $q<p$, then $q$ is not a prime factor of $2^{p}-1$. So $\mathrm{gcd}(2^{p}-1,q)=1$ which implies that $r$ is a factor of $2^{p}-1$. If $q>p$, suppose $\mathrm{gcd}(r,q)\neq1$ then $\mathrm{gcd}(r,q)=q$, so there exists a positive integer $k$ such that
 \begin{equation}\label{eq6}
 kp+1=q
 \end{equation}
 which is equivalent to $q-1=kp$. Because $2=\mathrm{gcd}(p-1,q-1)=\mathrm{gcd}(p-1,kp)=\mathrm{gcd}(p-1,k)$, so $k=3\mod4$. In Eq.\eqref{eq6},$k=3\mod4,p=1\mod4$ and $q=3\mod4$ which is impossible. So $\mathrm{gcd}(r,q)=1$. Then $\mathrm{gcd}(2^{p}-1,\frac{2^{N}-1}{2^{p}-1})=1$.

\item[(3)]$\mathrm{Ord}_{r}(2)=q$.

 Using the same method as that in the case $(2)$, we also have
 $\mathrm{gcd}(2^{q}-1,\frac{2^{N}-1}{2^{q}-1})=\mathrm{gcd}(2^{q}-1,p)=1$.
\end{itemize}
Combining the above discussion, we get
\begin{align*}
\mathrm{log}_{2}\big\lfloor \frac{2^{N}-1}{\mathrm{gcd}(s(2),2^{N}-1)}\big\rfloor &\geq \mathrm{log}_{2}\lfloor \frac{2^{N}-1}{\mathrm{gcd}(\mathrm{det}(A),2^{N}-1)}\rfloor\\
&\geq \mathrm{log}_{2}\lfloor \frac{2^{N}-1}{(2^{p}-1)(2^{q}-1)}\rfloor\geq N-p-q-1.
\end{align*}
\end{proof}

\begin{theorem}
Let $p$ and $q$ be two primes satisfying $p\equiv1\mod 4,q\equiv3\mod 4$ with $q-p=2$. Let $\{s_{i}\}_{i=0}^{N-1}$ be the Whiteman generalized cyclotomic sequence defined in Eq.(\ref{eq1}). Then the 2-adic complexity of $\{s_{i}\}_{i=0}^{N-1}$ is maximal, i.e.,
$$\phi_{2}(s)=N.$$
\end{theorem}
\begin{proof}
Let $r$ be a prime factor of $2^{N}-1$ and $\mathrm{Ord}_{r}(2)$ be the multiplicative order of 2 modulo $r$. Since $2^{N}\equiv1\mod r$, so $\mathrm{Ord}_{r}(2)|N$. Then we get $\mathrm{Ord}_{r}(2)\in\{pq,p,q\}$. By Fermat's little theorem, we know that $2^{r-1}=1\ \mod\ r$. Then $\mathrm{Ord}_{r}(2)|r-1$, therefore $r=k\mathrm{Ord}_{r}(2)+1$ where $k$ is a positive integer. From Lemma 7, we first calculate the value of $\mathrm{gcd}(\mathrm{det}(A)),2^{N}-1)$ for different cases.  For the case of $ord_{r}(2)=pq$, we have $\mathrm{gcd}(r,\mathrm{det}(A))=1$ by Theorem 1. Following, we only analyze $gcd(r,\mathrm{det}(A))$ for the $r$'s satisfying $\mathrm{Ord}_{r}(2)\in \{p,q\}$. For the reason that such $r$'s satisty $r>\mathrm{min}\{p,q\}$. Among the factors in $\mathrm{det}(A)$, $\frac{p-1}{4}$ and $\frac{1+q}{4}$ are less than $r$. So we only need to analyze $\mathrm{gcd}(r,\frac{pq-1}{2})$ and $\mathrm{gcd}(r,\Delta)$ for the case of $\mathrm{Ord}_{r}(2)\in \{p,q\}$.
\begin{itemize}
\item[(1)] $\mathrm{Ord}_{r}(2)=p$.\\

(i)$\mathrm{gcd}(r,\Delta)=1$.\\

If $\Delta=[\frac{(1+pq)^{2}}{16}+\frac{1+pq}{2}]$, suppose $r=kp+1$ is the factor of $\Delta$. Then
there exists an integer $t$ such that $\frac{(1+pq)^{2}}{16}+\frac{1+pq}{2}=t(kp+1)$, i.e.$\frac{(p+1)^{4}}{16}+\frac{(p+1)^{2}}{2}=t(kp+1)$ which is equivalent to $(p+1)^{4}+8(p+1)^{2}=16t(kp+1)$ i.e.$(p+1)^{2}[(p+1)^{2}+8]=16t(kp+1)$. By the fact that $r=kp+1$ is prime so $k>1$ and then $kp+1\nmid p+1$. So $kp+1\mid(p+1)^{2}+8$. Suppose
 \begin{equation}\label{eq7}
 (p+1)^{2}+8=t_{1}(kp+1)
 \end{equation}
 i.e.$p(p+2)+9=t_{1}kp+t_{1}$. If $t_{1}<p$, then $t_{1}=9$ and $9k=p+2=q$ which contradicts with $q$ is prime. So $t_{1}>p$ and then $t_{1}k<p+2$. So $k=1$ which contradicts with $r=kp+1$ is prime. So $\mathrm{gcd}(r,\Delta)=1$. \\
 If $\Delta=[\frac{(1+pq)^{2}}{16}+\frac{1-pq}{2}]$, suppose $\frac{(1+pq)^{2}}{16}+\frac{1-pq}{2}=t(kp+1)$ i.e.$(p+1)^{4}+8[1-p^{2}-2p]=16t(kp+1)$. That is $[(p+1)^{2}-4]^{2}=16t(kp+1)$. Then $kp+1|p+3$ or $p+1|p-1$ which is impossible. So we also have $\mathrm{gcd}(r,\Delta)=1$.\\

 (ii) $\mathrm{gcd}(r,\frac{pq-1}{2})=1$.\\

Suppose $\frac{pq-1}{2}=t(kp+1)$ i.e. $pq-1=2t(kp+1)$, that is to say,
\begin{equation}\label{eq8}
p(p+1)+(p-1)=2tkp+2t.
\end{equation}
If $2t<p$, in Eq. \eqref{eq8},then $2t=p-1$ and then $(p-1)kp=p(p+1)$ so $k(p-1)=p+1$ which is impossible. So $2t>p$ and $2tk<p+1$ which is impossible, too. So $ \mathrm{gcd}(r,\frac{(pq-1)}{2})=1$.\\

\item[(2)]$\mathrm{Ord}_{r}(2)=q$.\\

By the same method as that used in (1), we also get $\mathrm{gcd}(r,\Delta)=1$ and $\mathrm{gcd}(r,\frac{pq-1}{2})=1$. Furthermore,
$\mathrm{gcd}(r,\mathrm{det}(A))=1$.
\end{itemize}
By the above analysis, we have $\mathrm{gcd}(2^{N}-1,\mathrm{det}(A))=1$. Then

  $$N-1\geq \mathrm{log}_{2}\big\lfloor \frac{2^{N}-1}{\mathrm{gcd}(s(2),2^{N}-1)}+1\big\rfloor\geq \mathrm{log}_{2}\lfloor \frac{2^{N}-1}{\mathrm{gcd}(\mathrm{det}(A),2^{N}-1)}+1\rfloor=N.$$
\end{proof}

\begin{rem}
For the results obtained in this paper, we make the following two explanations:
\begin{itemize}
\item[(1)] For the result of Lemma \ref{mainlemma}, we have verified its correctness by Matlab and Mathematica programs through several examples.
\item[(2)] Through direct calculation using Mathematica programs, we get and list the exact values of 2-adic complexity and the lower bounds of the 2-adic complexity obtained in this paper for different cases of $(p,q)$ in the following table. From the table, we found that the exact values of the 2-adic complexity in all the following examples attain the maximal. So we conjecture that the 2-adic complexity of the sequences discussed in this paper may be really maximal for all possible parameters but we are not able to prove it in our ability. We also sincerely invite those readers who are interested in it to take part in this work.
\end{itemize}
\end{rem}
\begin{tabular}{cccccc}
  \hline
  $(p,q)$ & $\mathrm{2-adic}$ & $\mathrm{Lower \  bound}$ & $(p,q)$ &  $\mathrm{2-adic}$ & $\mathrm{Lower \  bound}$\\
\hline

  $(5,3)$ &  $15$              & $6$                     &  $(17,11)$ & $187$           &  $158$                   \\
  $(5,7)$ &  $35$              & $22$                    &  $(17,19)$ & $323$           &  $286$                 \\
  $(5,11)$&  $55$              & $38$                    &  $(17,23)$ & $391$           &  $350$                 \\
  $(13,11)$ & $143$            &  $118$                  &  $(17,31)$ & $527$           &  $478$                 \\
  $(13,23)$ & $299$            &  $262$                  &   $(17,43)$ & $731$          &  $670$                  \\

  \hline
\end{tabular}

\section{Summary}
In this paper, we study the 2-adic complexity of a class of balanced generalized cyclotomic sequences, which has been proved to have high linear complexity to resist linear attacks. The results of this paper show that it has also have large 2-adic complexity, i.e., it is larger than half of the period and is large enough to resist the attack of the RAA.


\begin{thebibliography}{}
\bibitem{Davis} P. J. Davis, "Circulant Matrices." New York, NY, USA: Chelsea, 1994.
\bibitem{Andrew Klapper} Klapper, A., Goresky, M.: Feedback shift registers, 2-adic span, and combiners with memory. Journal of Cryptology  10, 111-147 (1997).
\bibitem{Tian Tian} Tian, T., Qi, W.: 2-Adic complexity of binary $m$-sequences. IEEE Trans. Inform. Theory 56, 450-454 (2010).
\bibitem{Xionghai-1} Xiong, H., Qu, L., Li, C.: A new method to compute the 2-adic complexity of binary sequences. IEEE Trans. Inform. Theory 60, 2399-2406 (2014).
\bibitem{Hu Honggang} Hu, H.: Comments on  a new method to compute the 2-adic complexity of binary sequences.
 IEEE Trans. Inform. Theory 60, 5803-5804 (2014).
\bibitem{Xiong Hai-2}Xiong, H., Qu, L., Li, C.: 2-Adic complexity of binary sequences with interleaved structure. Finite Fields and Their Applications 33, 14-28 (2015).
\bibitem{Zeng X}Xiao Z, Zeng X and Sun Z. 2-Adic complexity of two classes of generalized cyclotomic binary sequences. International Journal of Foundations of Computer Science,  27 (07):879-893 (2016).
\bibitem{SunY}SunY, Wang Q and Yan T. A lower bound on the 2-adic complexity of the modified Jacobi sequence [J]. Cryptography and Communications, 21:1-13 (2019).
\bibitem{Sunyuhua-1}Sun Y, Wang Q, Wang Q and Yan T. A property of a class of Gaussian periods and its application [J]. IEICE Trans. Fundamentals, Vol. E101-A, No. 12, 2344-2351 (2018).
\bibitem{Hofer}Hofer R. and Winterhof A. On the 2-adic complexity of the two-prime generator [J]. IEEE Trans Inform. Theory, to appear.
\bibitem{Ding C}Ding C and Helleseth T. New generalized cyclotomy and its applications[J]. Finite Fields and Their Applications, 4(2):140-166 (1998).
\bibitem{Bai}Bai E, Liu X and Xiao G. Linear complexity of new generalized cyclotomic sequences of order two of length [J]. IEEE Trans. Inform. Theory,  51(5):1849-1853 .
\bibitem{Yan} Yan T. Construction and properties of Pseudo-random sequences. Xi'an, Xidian University, 2007.
\bibitem{Zhao}Zhao C, Sun Y and Yan T. Study on 2-adic complexity of a class of balanced generalized cuclotomic sequences[J].Journal of Cryptologic Research,6(4):455-462 (2019).

\end{thebibliography}
\end{document}